\documentclass{aic}

\usepackage[utf8]{inputenc}
\usepackage{tikz}
\usetikzlibrary{fit}
\usepackage{bbm}
\usepackage[shortlabels]{enumitem}
\usepackage{thmtools}
\usepackage{mathtools}

\newcommand{\pcb}{polynomially $\chi$-bounded}

\declaretheorem[name=Lemma, numberwithin = section]{lemma}
\declaretheorem[name=Theorem,sibling = lemma]{theorem}
\declaretheorem[name=Proposition, sibling=lemma]{proposition}

\declaretheorem[name=Corollary, sibling=lemma]{corollary}

\declaretheorem[name=Claim]{claim}

\aicAUTHORdetails{%
  title = {Bounded Twin-Width Graphs are Polynomially $\chi$-Bounded}, 
  author = {Romain Bourneuf and Stéphan Thomassé},
  plaintextauthor = {Romain Bourneuf, Stéphan Thomassé},
 }   

\aicEDITORdetails{%
   year={2025},
   number={2},
   received={28 March 2023},   
   published={13 February 2025},  
   doi={10.19086/aic.2025.2},      
}   

\begin{document}

\begin{frontmatter}[classification=text]


\author[rom]{Romain Bourneuf}
\author[stef]{Stéphan Thomassé\thanks{This work was supported by ANR Digraphs
ANR-19-CE48-0013
}}

\begin{abstract}
We show that every graph with twin-width $t$ has chromatic number $O(\omega ^{k_t})$ for some integer $k_t$, where $\omega$ denotes the clique number. This extends a quasi-polynomial bound from Pilipczuk and Soko\l{}owski and generalizes a result for bounded clique-width graphs by Bonamy and Pilipczuk. The proof uses the main ideas of the quasi-polynomial approach, with a different treatment of the decomposition tree. In particular, we identify two types of extensions of a class of graphs: the delayed-extension (which preserves polynomial $\chi$-boundedness) and the right-extension (which preserves polynomial $\chi$-boundedness under bounded twin-width condition). Our main result is that every bounded twin-width graph is a delayed extension of simpler classes of graphs, each expressed as a bounded union of right extensions of lower twin-width graphs.
\end{abstract}
\end{frontmatter}


\section{Introduction}

One of the first questions regarding the chromatic number $\chi (G)$ of a graph $G$ is how it compares to the clique number $\omega (G)$. Indeed, while $\omega (G)\leq \chi (G)$, early constructions by Blanche Descartes \cite{BD54}, Zykov \cite{ZY52} and Mycielski \cite{MY55} show that there are triangle-free graphs with arbitrarily large $\chi$. This question was also one of the key motivations for the introduction of the probabilistic method by Erd\H{o}s \cite{ER59}, to get randomized constructions of graphs with large girth and large chromatic number.

\subsection{Polynomial $\chi$-bounded classes}
A hereditary class $\cal C$ of graphs (or just \emph{class} in this paper, i.e. closed under induced subgraphs)  is \emph{$\chi$-bounded} if there exists a function $f$ such that $\chi(G)\leq f(\omega (G))$ for every $G\in \cal C$. For instance, the well-known class of \emph{perfect graphs} is the (hereditary) class of graphs such that $f(x)=x$. When $f$ can be chosen as a polynomial, the class $\cal C$ is 
\emph{polynomially $\chi$-bounded}. Let us also say that $\cal C$ is \emph{$k$-initially $\chi$-bounded} if $f(i)$ exists for all $i\leq k$. 
Recently, Carbonero, Hompe, Moore and Spirkl~\cite{CHMS23} showed that there are classes which are $2$-initially $\chi$-bounded but not $\chi$-bounded. 
This result was then extended to arbitrary values of $k$, independently by Bria\'nski, Davies and Walczak~\cite{BDW22} and by Girão, Illingworth, Powierski, Savery, Scott, Tamitegama and Tan~\cite{GIPSSTT24}. Notably, there are classes which are $\chi$-bounded but not polynomially $\chi$-bounded, see~\cite{BDW22}. The field is developing at a very fast pace: for a recent survey, see Scott and Seymour~\cite{SS20}.

Polynomial $\chi$-boundedness is one of the tamest behaviours a class can have. The most natural example of \pcb~class is the class of perfect graphs, which can be defined by forbidden induced subgraphs (odd holes and antiholes as the Strong Perfect Graph Theorem asserts~\cite{CRST06}) and also admit some structural decompositions. This leads to three main directions of research: Which forbidden induced subgraphs give polynomial $\chi$-boundedness? 
Which structural parameters yield polynomial $\chi$-boundedness? Which operations preserve $\chi$-boundedness? These three questions are intimately connected. The simplest graphs which are \pcb, cographs, are altogether $P_4$-free graphs, have clique-width at most 2 and are the closure under substitutions of graphs of size at most 2.

From the forbidden induced subgraph perspective, the Strong Perfect Graph Theorem, proved by Chudnovsky, Robertson, Seymour
and Thomas~\cite{CRST06} is definitely the masterpiece of the field. One of the most influential questions in the domain was proposed by Gy\'arf\'as and Sumner: is the class of graphs excluding some fixed induced tree $\chi$-bounded? This is true for paths, but the $\chi$-bounding function is not known to be polynomial. Scott, Seymour and Spirkl~\cite{SSS23} showed that sparse graphs (i.e. with no $K_{t,t}$ subgraph) excluding a tree are \pcb. Liu, Schroeder, Wang and Yu showed it for $t$-broom free graphs~\cite{LSWY21}. 
Davies and McCarty~\cite{DM21} proved that circle graphs are polynomially $\chi$-bounded.

From the point of view of operations preserving (polynomial) $\chi$-boundedness, the landscape is less developed. Even deciding if modules can be safely contracted requires a careful argument. To this end,
Chudnovsky, Penev, Scott, and Trotignon~\cite{CPST13} showed that if a class is \pcb, its closure under substitutions is also \pcb. It is also natural to wonder whether combining two closure operations each preserving $\chi$-boundedness still preserves $\chi$-boundedness. This is unfortunately not always the case as substitutions and 2-cuts lead to triangle-free graphs with arbitrarily large $\chi$~\cite{BBDGT23}. On the positive side, Dvo\v{r}\'ak and  Kr\'al’ also showed in~\cite{DK12} that closure by bounded rank cuts also  preserves $\chi$-boundedness. For this, they introduced a new tree-decomposition which is reminiscent of our delayed decomposition tree.

\subsection{Twin-width}
With the exception of VC-dimension, classes of graphs with bounded complexity parameter are usually $\chi$-bounded. This is the case for bounded tree-width, as it implies bounded degeneracy, but also the case for rank-width, as shown in~\cite{DK12}. Building on this result,
Bonamy and Pilipczuk~\cite{BP19} showed that graphs with bounded clique-width are \pcb.

Twin-width was introduced in~\cite{BKTW21} as a new structural complexity measure, capturing at the same time minor-closed classes, strict classes of permutations and bounded clique-width classes.
It was shown in~\cite{BGKTW21} that graphs with bounded twin-width are $\chi$-bounded. In the same paper, a polynomial bound was posed as an open problem, which would extend the polynomial $\chi$-boundedness of graphs with bounded clique-width~\cite{BP19}.

A natural step when trying to achieve polynomial bounds is to first look for quasi-polynomial ones, that is of order $n^{\log^c n}$. For instance, $P_5$-free graphs are quasi-polynomially $\chi$ bounded, see~\cite{SSS22}. In their breakthrough result, Pilipczuk and Soko\l{}owski~\cite{PS22} showed the following result:

\begin{theorem}\label{th:PS}
For every $t \in \mathbb{N}$ there is a constant $\gamma_t$ such that every graph with twin-width at most $t$ and clique number $\omega$ has chromatic number bounded by $2^{\gamma_t \log^{4t+3}\omega}$.
\end{theorem}

Reaching quasi-polynomiality often requires the right tools and definitions, and adding a little bit more of structure can sometimes save the logarithmic term.
This paper is no exception, as the fundamental idea (a reduction to lower twin-width), was already proposed by Pilipczuk and Soko\l{}owski~\cite{PS22}. Their twin-width reduction is based on a clever definition, $d$-almost mixed minors, which are particularly well-behaved for induced subgraphs. 
They use the fact that twin-width is functionally equivalent to finding a particular vertex ordering $v_1,\dots ,v_n$ for which the adjacency matrix cannot be partitioned into $k\times k$ mixed blocks, see~\cite{BKTW21}. Here mixed means that there are at least two distinct rows and 
two distinct columns. However, this definition is too constrained and they relaxed it by lifting the condition for the diagonal blocks ($k$-almost mixed minors). The difficult technical part of their approach is to partition $G$ into some restrictions $G[v_i,\dots,v_j]$ which do not ``contain" a $k-1$-almost mixed minor, in order to apply induction. The ``containment" notion is a very clever mix of edge-partition and vertex-quotient. This is the central reduction of their argument.
Our proof also uses this reduction and inserts it inside a different analysis of the decomposition tree. As a result we obtain:

\begin{theorem}\label{th:BT}
For every $t \in \mathbb{N}$ the class of graphs with twin-width at most $t$ is polynomially $\chi$-bounded. 
\end{theorem}

\subsection{Overview of the proof}
The first step is to consider a decomposition tree analogous to the substitution tree used for modules. The idea is simple: we first start with the partition $$\{\{v_1, \ldots, v_{\lceil n/2 \rceil}\}, \{v_{\lceil n/2 \rceil + 1}, \ldots, v_n\}\}$$ and we greedily continue to partition every part $B$ into modules with respect to the outside of $B$. If we reach a real module, we c
ut it in half and iterate. We now have a decomposition tree $T_d$ whose leaves are the vertices of $G$.

The next step is to observe that we can structure $T_d$ (i.e. associate a graph $g(x)$ to every internal node $x$) so that the information $(T,g)$ is enough to retrieve $G$. This is our \emph{delayed decomposition}. The key-fact is that if the class of graphs $g(x)$ is \pcb, one can derive polynomial $\chi$-boundedness for $G$. This does not involve twin-width and is a general decomposition method: delayed extensions preserve (polynomial) $\chi$-boundedness.

When $G$ is $k$-almost-mixed free, we can argue as in~\cite{PS22} that all $g(x)$ are ``simpler". Here again, we encapsulate the argument into a general framework, \emph{right extensions}, for which we prove that they preserve $\chi$-boundedness for all graphs, and polynomial $\chi$-boundedness for bounded twin-width graphs.

To sum up our approach: we show that $d$-almost mixed free graphs are delayed extensions of (vertex and edge unions of) right extensions of $d-1$-almost mixed free graphs and bounded $\chi$ graphs.

\subsection{Future directions}

In this paper, we stress the roles of two main operations: delayed extensions and right extensions, which apply to general graphs. Whether they can be used for other classical problems on $\chi$-boundedness is left for future research.

The fact that delayed extensions preserve polynomial $\chi$-boundedness directly results from the stability of polynomial $\chi$-boundedness by substitution. Since we crucially need it, we take a closer look at the proof from \cite{CPST13}. We slightly simplify the argument and get a better bound, again by tweaking the decomposition tree, but the whole argument is still surprisingly non trivial.

We also describe a new operation, the mixed extension, which appears naturally in a proof of the $\chi$-boundedness of bounded twin-width classes. A mixed subgraph of an ordered graph $G$ is obtained by only keeping the edges between mixed pairs of intervals, for some vertex-partition of $G$ into intervals. When $G$ has bounded twin-width (and the order is $d$-mixed free), every mixed subgraph has bounded chromatic number (using degeneracy via Marcus-Tardos). We show more generally that if the class $\cal C'$ of mixed subgraphs of a class $\cal C$ is $\chi$-bounded, then $\cal C$ is also $\chi$-bounded. Thus, mixed extensions preserve $\chi$-boundedness.

\section{First operation: Delayed extension}

Given a positive integer $s$, we denote by $[s]$ the set $\{1, 2, \ldots, s\}$. If $G$ is a graph, $\chi(G)$ denotes its chromatic number and $\omega(G)$ its clique number. When $X$ is a subset of vertices, we denote by $G[X]$ the graph induced by $G$ on $X$. We recall that $X\subseteq V(G)$ is a \emph{module} of $G$ if for every $y\in V(G)\setminus X$ we have all edges between $y$ and $X$, or no edge between $y$ and $X$.

We first show that every graph $G$ on vertex set $v_1, \ldots, v_n$ has a canonical decomposition tree $T_d$, called \emph{delayed decomposition tree}. Its leaves are the vertices of $G$, and every internal node $x$ of $T_d$ is labelled by a graph $G_x$ defined on the grandchildren of $x$. It is the analogue of the usual decomposition tree for modules, where every $G_x$ is defined on children instead of grandchildren (hence the "delayed"). 

The tree $T_d$ is defined via a sequence of refining partitions $P_0,P_1,\dots ,P_k$ of the vertex set $V$ of $G$ starting with the root vertex $P_0=\{V\}$ and ending on $P_k=\{\{v_1\}, \ldots, \{v_n\}\}$, the partition into singletons corresponding to the vertices. We now describe how to construct the refinement $P_{i+1}$ of $P_i=\{B_1, \ldots, B_m\}$, that is how a part $B_j$ is further refined. In this process, all parts will consist of consecutive vertices in the ordering $v_1, \ldots, v_n$.

For this, we define a partition $P(I)$ of an arbitrary interval of consecutive vertices $I=\{v_i,\dots, v_j\}$.

\begin{itemize}
    \item If $i=j$, it is not refined.
    \item If $I$ is a module, we divide it into two parts $\{v_i,\dots ,v_{\lfloor (i+j)/2\rfloor}\}$ and $\{v_{\lfloor (i+j)/2\rfloor+1},\dots, v_j\}$.
    \item If $I$ is not a module, we partition it into maximal subsets $S$ of consecutive vertices such that $S$ is a module in $G[(V\setminus I)\cup S]$. We call these parts \emph{local modules}.
   In other words $v_s,v_{s+1}\in I$ are in the same local module of $I$ if and only if they have the same neighbours in $V\setminus I$.
\end{itemize}

 To form the refinement $P_{i+1}$ of $P_i=\{B_1, \ldots, B_m\}$, we just refine every $B_j$ into $P(B_j)$. Observe that since $P_0=\{V\}$ is a module, we have $P_1 = \{\{v_1, \ldots, v_{\lceil n/2 \rceil}\}, \{v_{\lceil n/2 \rceil + 1}, \ldots, v_n\}\}$. We stop the process when $P_k=P_{k-1}$, which happens when all parts are singletons. For technical reasons (made clearer in the next definition) we keep the two identical partitions into singletons $P_{k-1},P_k$ instead of simply stopping at step $k-1$. The partition of a module into two (near) equal intervals is arbitrary, we could for instance partition $I$ into $\{v_i\}$ and $\{v_{i+1},\dots, v_j\}$.
 
We next consider the tree $T_d$ corresponding to this decomposition process, where the nodes at depth $i$ correspond to the parts of $P_i$, and the children of a node $x\in P_i$ are the parts $y\in P_{i+1}$ such that $y\subseteq x$ (we usually identify the nodes of $T_d$ to subsets of $V$). The leaves of $T_d$ are the vertices $v_i$ of $G$. Moreover, the parent of a leaf $v_i$ is $v_i$ since $P_{k-1}=P_k$. We now describe how to structure $T_d$ as a \emph{delayed tree decomposition} in order to encode the graph $G$. The crucial remark here is that if a node $x$ of $T_d$ has two grandchildren $y,z$ which are not siblings (we say that $y,z$ are \emph{cousins}, i.e. their parents are distinct), then we have all edges between $y$ and $z$, or no edge between them. In other words, cousins are modules with respect to each other.

From this observation, we define a function $g$ associating to every node $x\in T_d$ a graph $G_x:=g(x)$ whose vertex set is the set of  grandchildren of $x$ and such that $yz$ is an edge of $G_x$ if $y,z$ are cousins and there exists an edge between $y$ and $z$ in $G$ (and thus $y$ is fully joined to $z$). 

Given a pair $(T,g)$, now simply seen as a rooted tree $T$ in which the parent of every leaf only has one child, and each $g(x)$ is a graph on the grandchildren of $x$ (if any, otherwise $g(x)$ is empty), we define the \emph{realization} $R(T,g)$ as the graph such that: 
\begin{itemize}
    \item its vertex set is the set of leaves $L$ of $T$,
    \item two vertices $x,y\in L$ are joined by an edge if, given that $z$ is their closest ancestor in $T$ and $x',y'$ are the respective grandchildren of $z$ which are the ancestors of $x,y$, the edge $x'y'$ belongs to $g(z)$.
\end{itemize}

The crucial observation is that $G$ is equal to $R(T_d,g)$, hence every graph can be expressed as a delayed decomposition tree. Given a class of graphs $\cal C$, we denote by ${\cal C}_d$ the class of graphs $G$ admitting a delayed tree decomposition $(T_d,g)$ (for some enumeration of their vertex set) in which all graphs $g(x)$ belong to $\cal C$. We call ${\cal C}_d$ the \emph{delayed extension} of $\cal C$. It is not strictly speaking a closure since applying it twice can produce more graphs than applying it once.

The delayed extension is very similar to the \emph{substitution closure} ${\cal C}_s$ of $\cal C$, in which $G\in {\cal C}_s$ if all its induced \emph{prime} subgraphs $H$ (i.e. with no modules of size $k$ where $1<k<|V(H)|$) belong to $\cal C$. We have ${\cal C}_s\subseteq {\cal C}_d$, but we shall not need it. Conversely, we can express ${\cal C}_d$ in terms of ${\cal C}_s$. We invite the reader not familiar with tree-decompositions with respect to substitutions to read the Section~\ref{sec:subs} of this paper.

\begin{figure}[h]
\centering
\begin{tikzpicture}
  [level 1/.style={sibling distance=20mm},
   level 2/.style={sibling distance=7mm}]
  \node[circle, fill=brown]{}
    child{node[circle, fill=green]{} child {node[circle, draw=black](1){} child{node[circle, fill=black, label=below:A](6){}}} child {node[circle, draw=black](2){} child{node[circle, fill=black, label=below:B](7){}}} child {node[circle, draw=black](3){} child{node[circle, fill=black, label=below:C](8){}}}}
    child{node[circle, fill=cyan]{} child {node[circle, draw=black](4){} child{node[circle, fill=black, label=below:D](9){}}} child {node[circle, draw=black](5){} child{node[circle, fill=black, label=below:E](10){}}}};
    \draw[color=brown] (2) to[out=45,in=135]  (5);
    \draw[color=brown] (3) to (4);
    \draw[color=green] (6) to[out=45,in=135]  (8);
    \draw[color=green] (6) to (7);
    \draw[color=cyan] (9) to (10);
\end{tikzpicture}
\qquad \qquad
\begin{tikzpicture}
\node[circle, fill, label=above:A] (A) at ({72*0+90}:2){};
\node[circle, fill, label=left:C] (C) at ({72*1+90}:2){};
\node[circle, fill, label=left:D] (D) at ({72*2+90}:2){};
\node[circle, fill, label=right:E] (E) at ({72*3+90}:2){};
\node[circle, fill, label=right:B] (B) at ({72*4+90}:2){};
\draw[color=green] (A) -- (B);
\draw[color=brown] (E) -- (B);
\draw[color=cyan] (E) -- (D);
\draw[color=brown] (D) -- (C);
\draw[color=green] (A) -- (C);
\end{tikzpicture}
\caption{A delayed decomposition tree of $C_5$ and its corresponding realization. The edges of every $g(x)$ are drawn in the same color as $x$. Note that all $g(x)$ are cographs.}
\end{figure}
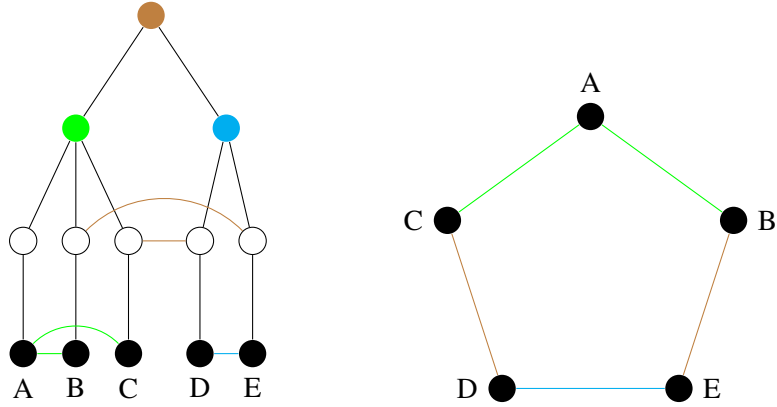

\begin{lemma}\label{lem:delayedpartition}
Every graph in ${\cal C}_d$ is the edge union of two graphs in ${\cal C}_s$.
\end{lemma}
\begin{proof}
Let $G:=R(T,g)$ be a graph such that $g(x)\in \cal C$ for all nodes of $T$. We consider a function $g_o$ (o for odd)  defined on $T$ such that  $g_o(x)=g(x)$ if $x$ has odd depth in $T$ and $g_o(x)$ is the edgeless graph (on the grandchildren of $x$) if $x$ has even depth. We define analogously $g_e$ in which $g_e(x)=g(x)$ if $x$ has even depth, and is edgeless otherwise. We define $G_o:=R(T,g_o)$ and $G_e:=R(T,g_e)$.

By construction, $G$ is the edge union $G_o\cup G_e$. We now show that $G_o$ is in ${\cal C}_s$. Observe that every node $x$ of $T$ with odd depth  (seen as a subset of vertices of $G$) is a module of $G_o$. Indeed, for every vertex $v\notin x$, the subset $x$ is a module of $G_o[x\cup v]$: this is by definition of $G_o$ if $v$ is not a descendant of a sibling of $x$, and if $v$ is a descendant of a sibling of $x$ there is no edge between $v$ and $x$ since their closest ancestor, the parent of $x$, has even depth.
In particular, if $H$ is a prime induced subgraph of $G_o$, we consider the deepest node $x$ with odd depth such that $V(H)\subseteq x$, and since $|V(H)\cap y|\leq 1$ for every grandchild $y$ of $x$ ($H$ is prime), we have that $H$ is an induced subgraph of $g(x)$, hence $H\in \cal C$. The same argument holds for $G_e$.
\end{proof}

We now recall the result of Chudnovsky, Penev, Scott and Trotignon~\cite{CPST13}.

\begin{theorem}\label{th:cpst}
If $\cal C$ is polynomially $\chi$-bounded with function $\omega^k$, then ${\cal C}_s$ is $\chi$-bounded by $\omega^{3k+11}$.
\end{theorem}

\begin{corollary}\label{cor:delay}
If $\cal C$ is polynomially $\chi$-bounded with function $\omega^k$, then ${\cal C}_d$ is $\chi$-bounded by $\omega^{6k+22}$.
\end{corollary}

\begin{proof} This directly follows from Lemma~\ref{lem:delayedpartition} and the fact that $\chi (G)\leq \chi (G_o)\chi (G_e)$.
\end{proof}

For the sake of self containment, we  provide in Section~\ref{sec:subs} a proof of Theorem~\ref{th:cpst} which is slightly simpler than the one of~\cite{CPST13} with the improved bound of $\omega^{2k+3}$ (hence giving $4k+6$ for delayed trees). It would be interesting to find the best bounds for Theorem~\ref{th:cpst}.

Thus, it suffices to focus on $\cal C$ if we want to show a polynomial $\chi$-bounding function for ${\cal C}_d$. The key to simplify a graph $G$ is an enumeration of its vertices for which the canonical delayed decomposition $(T_d,g)$ satisfies that all $g(x)$ are \emph{simpler} than $G$. This is better said for classes: A class ${\cal C}_0$ can be \emph{delayed} to ${\cal C}_1$ if ${\cal C}_0=({{\cal C}_1})_d$. For instance, cographs are the class $({{\cal G}_2})_d$ where ${\cal G}_2$ are the graphs of size at most 2 (indeed, it is $({{\cal G}_2})_s$). Let us say that a class ${\cal C}_0$ can be \emph{finitely delayed} to a class ${\cal C}_k$ if there exist ${\cal C}_0, {\cal C}_1,\dots ,{\cal C}_k$ such that ${\cal C}_i=({{\cal C}_{i+1}})_d$ for all $i=0,\dots ,k-1$. Also, we say that ${\cal C}_0$ has \emph{finite delay} if it can be finitely delayed to ${\cal G}_2$. By Corollary~\ref{cor:delay}, every class with finite delay is \pcb, as well as every class which can be finitely delayed to a \pcb~class. 

An exciting line of future research would be to explore which classes of graphs have finite delay to simpler classes. This could prove useful for understanding $\chi$-boundedness. Of course, it is a particular case of a more complex question asking if a graph can be edge-partitioned into two simpler graphs, but it has a great advantage: looking for arbitrary edge-partitions can be extremely complex due to the huge number of possibilities. Focusing instead on delayed decompositions only requires guessing the right vertex-ordering, which is a deeply explored field in graphs. To illustrate this, assuming that the vertex ordering of $G\in \cal C$ is a breadth first search and that the rule for cutting modules is to isolate the first vertex, note that $T_d$ will mimic the discovery order. In other words, if the induced subgraphs on the layers of the BFS belong to a  class simpler than $\cal C$, then $\cal C$ is a delayed extension of a simpler class. In this sense, delayed decompositions can generalize $\chi$-boundedness arguments based on layer decompositions.

Delayed decompositions also behave very well with twin-width, which has an alternative definition involving vertex orderings.

\section{Second operation: Right extension}

Our goal in this section is to define an extension of a class of graphs $\cal C$ which preserves $\chi$-boundedness, and even polynomial $\chi$-boundedness when the twin-width of $\mathcal{C}$ is bounded. Given a graph $G$, a \emph{right module partition} (RMP) is a partition $V_1, \ldots, V_k$ of the vertices of $G$ such that \begin{enumerate}
    \item Each $V_i$ is a stable set.
    \item For every $i < j$, $V_i$ is a module with respect to $V_j$ (i.e. $V_i$ is a module in $G[V_i\cup V_j]$).
\end{enumerate}

Note that every graph $G$ has a trivial RMP where each $V_i$ consists of a single vertex. Therefore, there should be some limitations to the definition of RMP. A first attempt is to consider a class of graphs $\mathcal{C}$, and insist that every induced subgraph intersecting every $V_j$ on at most one vertex (called a \emph{transversal}) belongs to $\mathcal{C}$. 
Unfortunately, even RMP with forests transversals are not $\chi$-bounded. To see this, consider $S_{n,2}$, the $n$-th shift graph, whose vertex set is $\{(i, j), 1 \leq i < j \leq n\}$ and such that there is an edge between $(i, j)$ and $(i', j')$ if and only if $j = i'$. Erd\H{o}s and Hajnal \cite{EH64} proved that the graphs $S_{n,2}$ are triangle-free and have unbounded chromatic number (this is a direct application of Ramsey theorem on pairs). However, the partition $(V_2, \ldots, V_n)$ where $V_j = \{(i, j), 1 \leq i < j\}$ for $2 \leq j \leq n$ is an RMP (with $V_1$ empty) such that the only neighbours of $(i,j)\in V_j$ in parts $V_k$ where $k<j$ are in $V_i$. Thus if $(i,j)$ belongs to a transversal, its degree is at most one with respect to the vertices in $V_k$ with $k < j$. Hence, all transversals are forests, while the graphs $S_{n,2}$ are not $\chi$-bounded.

For this reason, we introduce a stronger notion of RMP, meant to preserve $\chi$-boundedness.
    If $\mathcal{P} = (V_1, \ldots, V_k)$ is an RMP of a graph $G$, for every $1 \leq j_1 < j_2 < \ldots < j_\ell \leq k$ and every $W_{j_1} \subseteq V_{j_1}, \ldots, W_{j_\ell} \subseteq V_{j_\ell}$, all non-empty, we denote by $G/\{W_{j_1}, \ldots ,W_{j_\ell}\}$ the graph on vertex set $[\ell]$ such that there is an edge $ii'$ if and only if there is an edge between $W_{j_i}$ and $W_{j_{i'}}$. We call such a graph a \emph{transversal minor} of $(G, \mathcal{P})$.
    Given a class $\mathcal{C}$, an RMP such that all transversal minors are in $\mathcal{C}$ is called a \emph{$\mathcal{C}$-RMP}. The class of graphs $G$ admitting a $\mathcal{C}$-RMP is denoted by $RM(\mathcal{C})$ and is called the \emph{right extension} of $\cal C$.

\begin{figure}[h]
\centering
\begin{tikzpicture}
\foreach \x/\y/\z in {1/2/A, 1/3/B, 2/3/C, 1/4/D, 2/4/E, 3/4/F, 1/5/G, 2/5/H, 3/5/I, 4/5/J} {
  \node at (2.1*\y-2.1, 1.5*\x+1.5-1.5*\y) (\z) {\x, \y};
}
\node[draw,dotted,fit=(A)] (2) {};
\node[draw,dotted,fit=(B) (C)] (3) {};
\node[draw,dotted,fit=(D) (E) (F)] (4) {};
\node[draw,dotted,fit=(G) (H) (I) (J)] (5) {};
\draw (C) to [bend left = 70] (2.south);
\draw (E.west) to [bend left = 70] (2.south);
\draw (F) -- (3.east);
\draw (H.west) to [bend left = 70] (2.south);
\draw (I.west) to [bend left = 70] (3.east);
\draw (J) -- (4);
\end{tikzpicture}
\caption{The RMP for $S_{5, 2}$. Here, an edge means that we have all edges from the stable set on the left to the vertex on the right. Observe that every transversal is a forest, however we can form every graph on 4 vertices as a transversal minor.}
\end{figure}
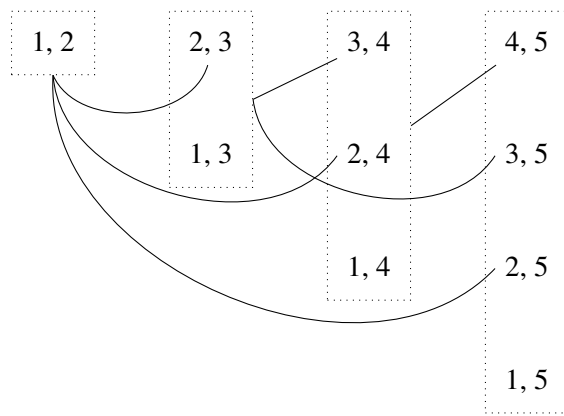

Our goal is now to show that right extension preserves $\chi$-boundedness. We will express our result in the language of ordered graphs (graphs with a total order on vertices). An RMP for an ordered graph must respect the order, that is the parts of the partition must consist of consecutive vertices. The following result is not used in the proof that bounded twin-width graphs are \pcb, but it could be useful in other context and its proof is similar to a later argument.  

\begin{proposition}\label{prop:chibound} If $\mathcal{C}$ is a  $\chi$-bounded class of ordered graphs, then $RM(\mathcal{C})$ is $\chi$-bounded.
\end{proposition}

The proof of Proposition~\ref{prop:chibound} is heavily based on the following lemma about the clique number. Given an ordered graph $G$ and an RMP ${\mathcal P}=(V_1, \ldots, V_k)$, we denote by $G/\cal P$ the ordered graph obtained by contracting all parts of $\cal P$, i.e. $G / \{V_1, \ldots, V_k\}$. This is a particular transversal minor of $G$, with the property that $\chi (G) \leq \chi (G/{\cal P})$. 
A class $\cal C$ of ordered graphs is \emph{$h$-free} if it does not contain some ordered graph on $h$ vertices.

\begin{lemma} There exists a function $\phi$ such that for every $h$-free class $\cal C$ and every ordered graph $G$ with a $\cal C$-right module partition $\cal P$, we have $\omega (G/{\cal P})\leq \phi (\omega (G),h)$
\end{lemma}

\begin{proof}
    The proof is by induction on $\omega :=\omega(G)$ and $h$. 
    If $h=1$ or $\omega =1$ then $G$ is edgeless, so we can set $\phi(x,1)=\phi(1,y)=1$. 
    Now, $\omega \geq 2$ and $h \geq 2$, and we assume that we proved the existence of $\phi(\omega - 1, h)$ and $\phi(\omega, h-1)$.
    We denote $\mathcal{P} = (V_1, \ldots, V_k)$ and consider an ordered graph $H$ on vertices $v_1,\dots ,v_h$ which is not in $\cal C$. Observe that we can restrict ourselves to the case where $G/{\cal P}$ is a clique, as we can only consider a maximal clique of $G/{\cal P}$.
    
    Thus, for every $i < j \in [k]$, there is an edge between $V_i$ and $V_j$.
    Let $D_k$ be a subset of $V_k$ of minimal size such that there is an edge between $V_i$ and $D_k$ for every $i < k$. By minimality of $D_k$, for every $d \in D_k$ there exists $i_d \in [k-1]$ such that there is an edge between $V_{i_d}$ and $d$ (and thus $d$ dominates $V_{i_d}$) but no edge between $V_{i_d}$ and $D_k \setminus \{d\}$. We consider two cases depending on the size of $D_k$. 
    
    \begin{enumerate}
        \item If $|D_k| > \phi(\omega, h-1)+1$, we show that we reach a contradiction. Select some $x\in D_k$ and consider the set $\cal P '$ of $|D_k|-1$ parts $V_{i_d}$ as defined above, except for the part $V_{i_x}$ which is not selected in $\cal P '$. Since $G/\cal P'$ is a clique of size at least $\phi(\omega, h-1)+1$, it follows by induction that $(G[\cup {\cal P'}],\cal P')$ contains all ordered graphs of size $h-1$ as transversal minors. In particular, the ordered graph $H'=H\setminus v_h$ is a transversal minor of $\cal P '$.
        If $v_h$ is isolated in $H$, we reach a contradiction since $H'\cup x$ is isomorphic to $H$ and is a transversal minor of $(G,\cal P)$. Otherwise, observe that one can extend $H'$ in all possible ways as a transversal minor of $\cal P$ by selecting some vertices in $D_k$. Indeed, every $V_{i_d}$ corresponds to a vertex $d$ in $D_k$ which is only joined to $V_{i_d}$. We can then select vertices in $D_k$ to extend $H'$ to $H$, a contradiction.
        
        \item If $|D_k| \leq \phi(\omega, h-1)+1$. For $d \in D_k$, let $S_d$ be the set of neighbours of $d$ in $G$. In particular, $\omega(G[S_d]) \leq \omega - 1$. Furthermore, $\mathcal{P}$ induces by restriction a $\mathcal{C}$-RMP $\mathcal{P}_d$ of $G[S_d]$. We then have $\omega(G[S_d]/\mathcal{P}_d) \leq \phi(\omega - 1, h)$. By taking the union over all $d \in D_k$, we deduce $\omega(G/\mathcal{P}) \leq \phi(\omega - 1, h) \cdot (\phi(\omega, h-1) +1)+ 1$ (the additional +1 stands for the last class $V_k$ which is not dominated).

    \end{enumerate}

Therefore, we can choose $\phi(\omega, h) = \phi(\omega - 1, h) \cdot (\phi(\omega, h-1) +1)+ 1$.
\end{proof}

We are now ready to prove Proposition~\ref{prop:chibound}. If $\cal C$ has a $\chi$-bounding function $f$, by the fact that the class of all graphs is not $\chi$-bounded, there is a graph $H$ of size $h$ which is not in $\cal C$. Consider now any graph $G$ in $RM(\cal C)$ with clique number $\omega$ and $\mathcal{C}$-RMP $\cal P$. We have $$\chi (G)\leq \chi (G/{\cal P})\leq f(\omega (G/{\cal P}))\leq f(\phi(\omega(G),h)).$$

Therefore the function $f(\phi(\omega(G),h))$ is $\chi$-bounding for $RM(\cal C)$. This approach does not provide a polynomial bound if the class $\cal C$ is polynomially $\chi$-bounded. We could not prove (or disprove) that polynomial $\chi$-boundedness is preserved by RMP. However, this is the case when the twin-width of $\mathcal{C}$ is bounded, as shown in the following section. 
\section{Twin-width and almost-mixed minors}
  
  We recall here some definitions and results related to twin-width. For a more intuitive and pedestrian introduction, see~\cite{BKTW21}. We adopt here the matrix point of view of twin-width, where every graph $G$ is represented via its symmetric adjacency matrix $(a_{u,v})$ where $u,v$ are over couples of vertices. The entry $a_{u,v}$ is 1 if $uv$ is an edge, 0 if $uv$ is not an edge, and $*$ if $u=v$. The addition of the $*$ symbol slightly simplifies some technicalities, but is not necessary for the argument.
  
    A $01*$-matrix is \emph{horizontal} if all its rows are constant. It is \emph{vertical} if all its columns are constant. It is \emph{constant} if it is both horizontal and vertical. It is \emph{mixed} if it is neither horizontal nor vertical, or if it has at least 2 rows and 2 columns and contains a $*$ entry.
A \emph{corner} in a matrix $M$ is a mixed $2 \times 2$ submatrix of $M$.

\begin{lemma}[\cite{BKTW21}]\label{lem:corner}
 A matrix is mixed if and only if it contains a corner.
\end{lemma}

   Let $M$ be a matrix. 
    A \emph{row partition} of $M$ is a partition of the rows of $M$ in which each part of the partition consists of consecutive rows. We define \emph{column partitions} in a similar way. 
    A \emph{division} $\mathcal{D}$ of $M$ is a pair $(\mathcal{R}, \mathcal{C})$ where $\mathcal{R}$ is a row partition of $M$ and $\mathcal{C}$ is a column partition of $M$. If $\mathcal{R}$ and $\mathcal{C}$ both have the same number of parts, say $k$, we say that $(\mathcal{R}, \mathcal{C})$ is a \emph{$k$-division} of $M$. In this case, we index the row blocks and the column blocks of $\mathcal{D}$ with integers from $[k]$ in the natural order of the blocks. 
    If $\mathcal{D}$ is a $k$-division of $M$, for $i, j \in [k]$, we denote by $\mathcal{D}[i, j]$ the intersection of the $i$-th row block with the $j$-th column block, which we call a \emph{zone} of $\mathcal{D}$.  If $R_i\in \mathcal{R}$ and $C_j\in\mathcal{C}$, we also adopt the notation $[R_i,C_j]$ for the zone $\mathcal{D}[i, j]$. It is a contiguous submatrix of $M$. 
    
    We say that a zone of a matrix is \emph{mixed} if it is mixed as a submatrix.
    If $M$ is symmetric, we say that a division $(\mathcal{R}, \mathcal{C})$ is \emph{symmetric} if $\mathcal{R}$ and $\mathcal{C}$ partition rows and columns in the same way (i.e. $\mathcal{R}$ is the transpose of $\mathcal{C}$).
    
    Let $M$ be a matrix and $\mathcal{D}$ be a $d$-division of $M$. We say that $\mathcal{D}$ is a \emph{$d$-mixed minor} if each zone of $\mathcal{D}$ is mixed. If $M$ does not have any $d$-mixed minor, we say that $M$ is \emph{$d$-mixed free}. The twin-width parameter and mixed-minor freeness are functionally equivalent. In particular, the following was shown in \cite{BKTW21}:

    \begin{lemma}\label{lem:mixedtww}
    If a graph $G$ has twin-width at most $d$, it has a vertex ordering for which the adjacency matrix of $G$ is $f_d$-mixed free for some constant $f_d$.
    \end{lemma}

    The next result is a direct consequence of the Marcus-Tardos theorem, see \cite{MT04}:

\begin{theorem}\label{th:marcustardos}
    For every positive integer $d$, there is a constant $mt_d$ such that for every $d$-mixed free matrix $M$ and every $k$-division of $M$, the number of mixed zones is at most $mt_d \cdot k$
\end{theorem}

As we will perform some operations on our graphs (such as deleting edges and contracting subsets of vertices), we show in the next results how mixed zones are affected. Let $M$ be a $01*$-matrix (not necessarily an adjacency matrix) with exactly two row blocks $R,R'$ and two columns blocks $C,C'$, each of size at least two. 

\begin{lemma}\label{lem:4corner}
If all four zones of $M$ are mixed, there is a corner intersecting all zones.
\end{lemma}

\begin{proof}
This is clear if $M$ contains a $*$. Consider a non constant row $r$ in $[R,C]$ and a non constant column $c'$ in $[R',C']$. Let $e$ be the entry $(r,c')$. Let $c\in C$ such that $e\neq (r,c)$ and $r'\in R'$ such that $e\neq (r',c')$. Now $\{r,r'\},\{c,c'\}$ is a corner.
\end{proof}

The \emph{contraction} $M'=M/\{R,R';C,C'\}$ is the $2\times 2$-matrix obtained by keeping a single value $x$ for each of the 4 zones, with the following rule: $x$ is the maximum value of the zone according to the order $0<1<*$. Thus, we get $*$ as soon as there exists a $*$, and we get $0$ only if the zone is full $0$.

\begin{lemma}\label{lem:contract}
If $M'$ is mixed, then $M$ is mixed.
\end{lemma}

\begin{proof}
This is clear if $M'$ (thus equivalently $M$) contains a $*$. Let us assume by contrapositive that $M$ is non-mixed. If $M$ is vertical (resp. horizontal), then observe that $M'$ is also vertical (resp. horizontal).
\end{proof}
 
 We keep the same notations as before, and assume moreover that none of the four zones of $M$ is mixed (in particular $M$ has no value $*$). The \emph{horizontal-deletion} $M_H$ of $M$ is the matrix obtained from $M$ by setting all values to 0 in each zone which is not vertical (or equivalently each zone which is horizontal and non-constant). We similarly define the \emph{vertical-deletion} $M_V$.
 
\begin{lemma} \label{lem:deletion}
If $M_H$ is mixed, then $M$ is mixed.
\end{lemma}

\begin{proof}
 Let us assume by contrapositive that $M$ is non-mixed. By assumption, there is no $*$ in  $M_H$. If $M$ is vertical, then observe that $M_H$ is also vertical. Now if $M$ is horizontal, note that the zones $R,C$ and $R,C'$ are either both set to 0 (if they are not vertical), or both left as in $M$. In both cases the rows of $R$ in $M_H$ are constant. The same applies to $R'$, so $M_H$ is horizontal. 
\end{proof}

Here is the key-definition of \cite{PS22}.   A $d$-division $\mathcal{D}$ of a matrix $M$ is a \emph{$d$-almost mixed minor} if for every $i \neq j \in [d]$, the zone $\mathcal{D}[i, j]$ is mixed. If $M$ does not have any $d$-almost mixed minor, we say that $M$ is \emph{$d$-almost mixed free}. By extension, a graph is \emph{$d$-almost mixed free} if we can order its vertices in such a way that its adjacency matrix is $d$-almost mixed free.

Observe that every $d$-almost mixed free matrix is also $d$-mixed free. Conversely, every $d$-mixed free matrix is also $2d$-almost mixed free. Indeed, if $M$ has a $2d$-almost mixed minor, then merging the first $d+1$ row blocks, and the last $d+1$ column blocks gives a $d$-mixed minor of $M$. Note that every submatrix of a $d$-(almost) mixed free matrix is also $d$-(almost) mixed free, hence every subgraph of a $d$-almost mixed free graph is also $d$-almost mixed free. 

Let $G$ be a graph with an RMP $\mathcal{P} = (V_1, \ldots, V_k)$. We say that $(G,\mathcal{P})$ is \emph{$d$-almost mixed free}, if for every coarsening $\cal P'$ of $\cal P$ into $d$ parts $(V'_1, \ldots, V'_d)$, where each $V'_i$ consists of consecutive parts of $\cal P$, some zone $[V'_i,V'_j]$, where $i\neq j$, is not mixed in the adjacency matrix of $G$. Note that we only speak here of restrictions on symmetric divisions of $G$, which encompass much larger classes than bounded twin-width. 

\begin{lemma} \label{lem:clique-quotient}
If $(G,\cal P)$ is $d$-almost mixed free, then  $\omega (G/{\cal P})\leq \omega(G)^d$.
\end{lemma} 

\begin{proof}
We write $\omega:=\omega(G)$ and denote $\mathcal{P} = (V_1, \ldots, V_k)$. Let $\phi$ be such that $\omega(G/{\cal P}) \leq \phi(\omega, d)$. We have $\phi(\cdot, 1)=0$ (empty graph) and $\phi(1, \cdot) = 1$ (edgeless graph). We assume $\omega\geq 2$ and $d\geq 2$ and show that $\phi(\omega, d)= \phi(\omega-1, d)+\phi(\omega, d-1)+1$ upper bounds $\omega(G/{\cal P})$. We can restrict ourselves to a maximal clique of $G/{\cal P}$, so we can assume that there is an edge between $V_i$ and $V_j$ whenever $i < j \in [k]$.    

Let us consider the smallest $\ell$ such that $\omega(G[V_1\cup \dots \cup V_{\ell}])=\omega$. Denote by $Y$ the set $V_1\cup \dots \cup V_{\ell}$, and consider any $V_i$ where $i\geq \ell +1$. Note that $Y$ is not a module with respect to $V_i$, as some vertex in $V_i$ would dominate $Y$, hence forming a clique of size $\omega +1$ in $G$. Conversely, if $V_i$ is a module with respect to $Y$, since $\cal P$ is an RMP, and there exists an edge between all pairs of parts, $V_i$ would dominate $Y$, with the same contradiction.

Consider the graph $G'=G[V_{\ell+1}\cup \dots \cup V_{k}]$ and its RMP $\mathcal{P'} = (V_{\ell +1}, \ldots, V_k)$. We claim that $(G',\cal P')$ is $d-1$-almost mixed free, otherwise any $d-1$-almost mixed minor coarsening $(V'_1, \ldots, V'_{d-1})$ of $\mathcal{P'}$ could be extended to the $d$-almost mixed minor $(Y,V'_1, \ldots, V'_{d-1})$ of $(G,\cal P)$. 

Thus $k=\omega (G/\cal P)$ satisfies by induction that $k\leq \ell + \phi(\omega, d-1)$. And since the first $\ell -1$ parts do not contain a clique of size $\omega$, we have $k \leq \phi(\omega - 1, d) + 1 + \phi(\omega, d-1)= \phi(\omega, d)$. Setting  $\psi(\cdot,\cdot)=\phi(\cdot,\cdot)-1$, we have that $\psi (\omega,d)=\psi(\omega - 1, d) + \psi(\omega, d-1)$. Moreover, we both have $\psi (\omega,1)=-1$ and $\psi (1,d)=0$, so $\psi (\omega,d)\leq \binom{\omega + d - 2}{d-1}\leq \omega ^{d-1}$.
\end{proof}

\begin{proposition} \label{prop:RMPpcb} Let $\cal C$ be a class of graphs with polynomial $\chi$-bounding function $f(x)=x^c$. If $\cal P$ is a $\cal C$-RMP of $G$ such that $(G,\cal P)$ is $d$-almost mixed free, then $\chi (G)\leq \omega ^{cd}$.
\end{proposition}

\begin{proof}
   Since $\omega(G/{\cal P}) \leq \omega(G)^{d}$ by Lemma \ref{lem:clique-quotient}, we have $\chi(G)\leq \chi(G/{\cal P})\leq \omega(G/{\cal P})^c\leq \omega(G)^{cd}$.
\end{proof}

\section{Polynomial $\chi$-boundedness of bounded twin-width graphs}

By Lemma~\ref{lem:mixedtww}, in order to prove our main result, Theorem~\ref{th:BT}, we just have to show that the class of $d$-mixed free graphs is polynomially $\chi$-bounded. Since mixed freeness is functionally
equivalent to almost mixed freeness, we only consider this last notion.

To prove that a class $\cal C$ is \pcb, a strategy is to show that every graph $G$ of $\cal C$ has a vertex-partition or an edge-partition into a bounded number of graphs, each of them belonging to some known \pcb~class $\cal C '$. We will use this argument several times here.

\begin{theorem}\label{thm:damf-pcb} The class of $d$-almost mixed free graphs is \pcb.
\end{theorem}

\begin{proof}
The proof is by induction on $d$. For $d=2$, if $G$ has a 2-almost mixed free adjacency matrix, $G$ is a cograph (see \cite{PS22}), hence $G$ is perfect so the property holds. Now, let $d \geq 3$ and consider a graph $G$, which is $d$-almost mixed free with respect to the vertex ordering $v_1,\ldots ,v_n$. We first partition the set of vertices $V' =\{v_2,\dots, v_{n-1}\}$ into four subsets $V'_{00},V'_{01},V'_{10},V'_{11}$ according to their neighbourhood in $\{v_1,v_n\}$. For instance, $V'_{01}=(V'\setminus N(v_1))\cap N(v_n)$. It suffices to show that all four graphs $G[V'_{ij}]$ belong to some \pcb~class. Hence, without loss of generality, we can assume that $V'$ is a module in $G$. 
    We now consider the delayed decomposition tree $(T_d,g)$ of $G':=G[V']$ and consider the class $\cal C$ containing all $g(x)$ and their induced subgraphs. By Corollary~\ref{cor:delay}, we just have to show that $\cal C$ is \pcb. The graphs $g(x)$ are obtained by starting from an interval $I=\{v_s,\dots, v_t\}$ of vertices of $G'$, partitioning it into local modules $L_1,\dots, L_k$, and then partitioning each local module into local submodules. Let $H$ be the graph we obtain from $G[I]$ by removing all edges inside all local modules $L_i$. Note that $H$ can also be obtained by substituting the vertices of $g(x)$ by stable sets. This does not change $\chi$ and $\omega$. Therefore, to show that $\mathcal{C}$ is \pcb, it suffices to prove the following.

\begin{claim} Let $I=\{v_s,\dots, v_t\}$ be any interval of vertices of $G'$. Consider its partition into local modules $L_1,\dots, L_k$ and denote by $H$ the graph on vertex set $I$ obtained from $G[I]$ by deleting the edges inside the local modules. Then, $H$ is a bounded (vertex and edge) union of graphs from hereditary \pcb~classes.
\end{claim}

Note that the claim holds when $I$ is a non trivial module since it is cut into two parts. Indeed, in this case $H$ is bipartite and thus $\chi$-bounded by 2. Assume now that we have local modules $L_1,\dots ,L_k$, and consider all pairs $i<j$ such that $G[L_i,L_j]$ is mixed (call \emph{mixed pairs}). The $L_i$'s form a $k$-division of the adjacency matrix of $G[I]$, which is $d$-mixed free. Thus, by Theorem \ref{th:marcustardos}, the graph $R$ on vertex set $[k]$ whose edges correspond to mixed pairs has at most $\frac{mt_d}{2} \cdot k$ edges. In particular, $R$ can be vertex-colored into $(mt_d+1)$ classes. In other words, one can partition the set of local modules into $(mt_d+1)$ subsets, in which local modules are not pairwise mixed. We denote by $L'$ such a subset of local modules. To prove our claim, we just have to show that $H':=H[L']$ belongs to a \pcb~class of graphs.

Observe that for every $i<j$ and $L_i,L_j\in L'$, we have that $L_i$ is a module in $H[L_i\cup L_j]$ (denoted $L_i\rightarrow L_j$) or $L_j$ is a module in $H[L_i\cup L_j]$, that is $L_j\rightarrow L_i$. Note that if we both have $L_j\rightarrow L_i$ and $L_i\rightarrow L_j$, we have all edges or no edge between $L_i$ and $L_j$. We now define two subgraphs $H'_{\rightarrow}$ and $H'_{\leftarrow}$ of $H'$: in $H'_{\rightarrow}$ we only keep the edges of $H'$ between pairs $L_i\rightarrow L_j$ where $i<j$, and in $H'_{\leftarrow}$ we only keep the edges of $H'$ between pairs $L_i\leftarrow L_j$ where $i<j$.
Note that $H'=H'_{\rightarrow}\cup H'_{\leftarrow}$, and thus we just have to show that (for instance) $H'_{\rightarrow}$ belongs to a \pcb~class of graphs.

The graph $H'_{\rightarrow}$ with the partition $L'$ is a right module partition. Note that the same holds for $H'_{\leftarrow}$ if we reverse the order of the local modules. 
We further partition $H'_{\rightarrow}$: let us say that a local module $L_i$ is \emph{left} if $i>1$ and there is a vertex $v_j$ among $v_1,v_2,\dots,v_{s-1}$ (i.e. to the left of $I$) which distinguishes $L_{i-1}$ from $L_i$. Precisely, $v_j$ is not joined in the same way to the last vertex of $L_{i-1}$ and to the first of $L_i$. If $L_i$ (with $i>1$) is not left, then it is \emph{right} (and indeed some vertex $v_j$ to the right of $I$ distinguishes $L_{i-1}$ from $L_i$). We neglect $L_1$ in this definition (it only adds 1 to the chromatic number of $H'$). We now partition $L'$ into $L'_{ri}$ containing all right local modules $L_i$ of $L'$, and $L'_{le}$ containing the left local modules. Again, by vertex partition, we just have to show that the RMP $H'_{\rightarrow,ri}$ which is the induced restriction of $H'_{\rightarrow}$ to $L'_{ri}$ is \pcb. To apply Proposition~\ref{prop:RMPpcb}, we first show that the transversal minors of  $(H'_{\rightarrow,ri},L'_{ri})$ are $d-1$-almost mixed free, which by induction implies that they are \pcb. Then we argue that $(H'_{\rightarrow, ri},L'_{ri})$ has no large almost mixed minor. It suffices here to show that it is $2d$-almost mixed free. These are our last results.

\begin{claim}[\cite{PS22}]\label{claim:twwreduction}
    Every transversal minor of  $(H'_{\rightarrow,ri},L'_{ri})$ is $d-1$-almost mixed free.
\end{claim}

\begin{proof}
Assume for contradiction that we can find a sequence of local modules $L'_1,\dots ,L'_t$ in $L'_{ri}$, each of them containing a non empty subset of vertices $W_1,\dots ,W_t$, such that the graph $Q=H'_{\rightarrow,ri}/\{W_1,\dots ,W_t\}$ has a $d-1$-almost mixed minor. The vertices of $Q$ are denoted $W=\{w_1,\dots ,w_t\}$, where $W_i$ is contracted to $w_i$. Moreover, there exist two partitions of $W$ into consecutive blocks of vertices $(R_1,\dots ,R_{d-1})$ and $(C_1,\dots ,C_{d-1})$ such that $Q$ is mixed on the zone $[R_i,C_j]$ whenever $i\neq j$ (thus all $R_i$ and $C_j$ have size at least 2). We now show how to "lift" these partitions to $G$ in order to get a contradiction.

Consider any partition ${\cal R}'=(R'_1,\dots ,R'_{d-1})$ of $I$ (where parts consist of consecutive local modules) satisfying that $w_i\in R_j$ implies $L'_i \subseteq R'_j$. Similarly ${\cal C}'=(C'_1,\dots ,C'_{d-1})$ partitions $I$ and $w_i\in C_j$ implies $L'_i \subseteq C'_j$. We now extend the partitions ${\cal R}',{\cal C}'$ of $I$ to the whole vertex set $V$ of $G$ by first setting $R'_1:=R'_1\cup \{v_1,\dots ,v_{s-1}\}$ and $C'_1:=C'_1\cup \{v_1,\dots ,v_{s-1}\}$, and then adding a new part $\{v_{t+1},\dots ,v_{n}\}=R'_d=C'_d$ to both ${\cal R}'$ and ${\cal C}'$. These new partitions are called ${\cal R}$ and ${\cal C}$. Observe that if we were working with $H'_{\rightarrow,le}$, we would have added the part $R'_0=C'_0=\{v_1,\dots ,v_{s-1}\}$ to both ${\cal R}'$ and ${\cal C}'$ and extended the parts $R'_{d-1}$ and $C'_{d-1}$ by adding $\{v_{t+1},\dots ,v_{n}\}$.

We now show that ${\cal R}$ and ${\cal C}$ form a $d$-almost mixed minor for $G$, which will be our contradiction. We need to focus on two points: the added parts $R'_d,C'_d$ should be mixed with respect to the others, and the original mixed zones $[R_i,C_j]$ of $Q$ should yield mixed zones $[R'_i,C'_j]$ of $G$. We separate the two arguments:
\begin{itemize}
    \item Consider first some zone $[R'_i,C'_d]$ where $i<d$ (similar argument for $[R'_d,C'_i]$). By the fact that $R_i$ contains two vertices $w_a,w_b$, the part $R'_i$ contains the right local modules $L'_a,L'_b$ (where $a<b$). Focus now on the vertex $v_j$ which is the first vertex of $L'_b$, and note that $v_{j-1}\in R'_i$ since $L'_a\subseteq R'_i$. Since $L'_b$ is a right local module, there exists $v_k$, where $k>t$ such that $v_k$ is differently joined to $v_{j-1}$ and $v_j$. Recall that $v_n$ is joined in the same way to $v_{j-1}$ and $v_j$. Since $v_{j-1},v_j\in R'_i$ and $v_k,v_n\in C'_d$, they witness the fact that $[R'_i,C'_d]$ is mixed. 
    \item Now, consider any zone $[R'_i,C'_j]$ where $i,j<d$ and $i\neq j$. If the zone $[R_i,C_j]$ contains a $*$ (i.e. some $w_a$ both belongs to $R_i$ and $C_j$), then $[R'_i,C'_j]$ also contains a $*$. Otherwise, by Lemma \ref{lem:corner}, $R_i$ contains two vertices $w_a,w_b$ and $C_j$ contains two vertices $w_c,w_d$ such that $\{w_a,w_b\},\{w_c,w_d\}$ is a corner. Moreover, since there is no $*$ value, we have $a<b<c<d$ or $c<d<a<b$. Without loss of generality, we assume $a<b<c<d$. By Lemma~\ref{lem:contract}, the restriction of the adjacency matrix of $H'_{\rightarrow,ri}$ on $[W_a\cup W_b,W_c\cup W_d]$ is mixed since its contraction is the corner $\{w_a,w_b\},\{w_c,w_d\}$. So the submatrix $[L'_a\cup L'_b,L'_c\cup L'_d]$ is also mixed.
    By definition of $H'_{\rightarrow,ri}$, if $L'_a$ (or $L'_b$) is not a module with respect to $L'_c$ (or to $L'_d$), then the zone $[L'_a,L'_c]$ is set to 0. In other words, the adjacency matrix of $H'_{\rightarrow,ri}$ restricted to $[L'_a\cup L'_b,L'_c\cup L'_d]$, is the horizontal-deletion of the adjacency matrix of $G$. Thus the zone $[R'_i,C'_j]$ is mixed by Lemma~\ref{lem:deletion}.
\end{itemize}
\end{proof}
\begin{claim}
    The pair $(H'_{\rightarrow,ri},L'_{ri})$ is $2d$-almost mixed free.
\end{claim}

\begin{proof}
    Assume for contradiction that we can find a coarsening $W'_1,\dots ,W'_{2d}$ of $L'_{ri}$ which forms a $2d$-almost mixed minor of $H'_{\rightarrow,ri}$. We now set $W_i=W'_{2i-1}\cup W'_{2i}$ for all $i=1,\dots ,d$. 
    By Lemma~\ref{lem:corner} every (mixed) zone $[W_i,W_j]$ with $i\neq j$ of $H'_{\rightarrow,ri}$ contains a corner $\{w_a,w_b\},\{w_c,w_d\}$.  By Lemma~\ref{lem:4corner}, we can assume that $w_a,w_b,w_c,w_d$ belong respectively to $W'_{2i-1},W'_{2i},W'_{2j-1},W'_{2j}$, hence to respective distinct local modules $L'_a,L'_b,L'_c,L'_d$.
    We assume without loss of generality that $i<j$, and thus $a<b<c<d$. 
    By definition of $H'_{\rightarrow,ri}$, if $L'_a$ (or $L'_b$) is not a module with respect to $L'_c$ (or to $L'_d$), then the zone $[L'_a,L'_c]$ is set to 0. In other words, the adjacency matrix of $H'_{\rightarrow,ri}$ restricted to $[L'_a\cup L'_b,L'_c\cup L'_d]$, is the horizontal-deletion of the adjacency matrix of $G$. Hence the zone $[W_i,W_j]$ is mixed in $G$ by Lemma~\ref{lem:deletion}. Thus $G$ restricted to $W_1,\dots ,W_d$ has a $d$-almost mixed minor, a contradiction.
\end{proof}
We can conclude the proof of Theorem~\ref{thm:damf-pcb} by using the two previous Claims and Proposition~\ref{prop:RMPpcb}.
\end{proof}

Theorem~\ref{th:BT} now follows from Lemma~\ref{lem:mixedtww} and Theorem~\ref{thm:damf-pcb}. From the previous proof, the order of magnitude for the $\chi$-bounding function of graphs with no $d$-almost mixed minor is $\omega^{d^{O(d)}}$.

\section{More on the $\chi$-boundedness of substitution-closure}\label{sec:subs}

The goal of this section is to provide a proof of the following theorem by slightly modifying the original argument of Chudnovsky, Penev, Scott and Trotignon~\cite{CPST13}.

\begin{theorem}\label{th:chibordec}
If $\cal C$ is hereditary and \pcb~with function $\chi \leq \omega^k$, then ${\cal C}_s$ is \pcb~with function $\chi \leq \omega^{2k+3}$
\end{theorem}

Given a class of graphs $\cal C$, a \emph{$\cal C$-tree-decomposition} is a pair $(T,g)$ in which $T$ is a rooted tree and $g$ is a function associating to every internal node $x$ of $T$ a graph $g(x)\in \cal C$ whose vertices are the children of $x$. The \emph{realization} $R(T,g)$ is the graph such that: 
\begin{itemize}
    \item its vertex set is the set of leaves $L$ of $T$
    \item two vertices $x,y\in L$ are joined by an edge if, given that $z$ is their closest ancestor in $T$ and $x',y'$ are the respective children of $z$ which are the ancestors of $x,y$, the edge $x'y'$ belongs to $g(z)$.
\end{itemize}

Given a class $\cal C$, we denote by ${\cal C}_s$ the class of all $R(T,g)$ where $(T,g)$ is a $\cal C$-tree-decomposition. For instance, if $\cal C$ is the class of cliques and independent sets, then  ${\cal C}_s$ is the class of cographs.
We say that $(T,g)$ is \emph{independent} if whenever $x, y$ are internal nodes of $T$ with the same parent $z$, then $xy$ is not an edge of $g(z)$. We denote by ${\cal C}_i$ the class of all $R(T,g)$ where $(T,g)$ is an independent $\cal C$-tree-decomposition. \\
Let $\cal C$ be a class of graphs, and $(T, g)$ be a $\cal C$-tree-decomposition. 
Let $u$ be any node of $T$ which is not the root, and let $v$ be the parent of $u$. We say that $u$ is \emph{isolated} in $T$ if it is isolated in $g(v)$. \\
The \emph{depth} $d(x)$ of a node $x$ of $T$ is the number of strict ancestors of $x$ in $T$ that are not isolated.
The \emph{depth} $d(T, g)$ of $(T, g)$ is the maximum depth of a leaf of $T$. 
Finally, if $G \in \mathcal{C}_s$, the \emph{depth} of $G$ is the minimum depth of a $\cal C$-tree-decomposition $(T, g)$ such that $G = R(T, g)$.

\begin{lemma}\label{lem:depth}
For every class $\cal C$, if $(T, g)$ is a $\mathcal{C}$-tree-decomposition of depth $d$, then $\omega(R(T, g)) \geq d + 1$.
\end{lemma}

\begin{proof} We consider a leaf $x$ with depth $d$. We denote by $y_1,\dots, y_d$ the non-isolated ancestors of $x$, and by $z_1,\dots, z_d$ their respective parents. We also denote by $w_1,\dots, w_d$ some respective children of $z_1,\dots, z_d$ such that $y_iw_i$ is an edge in $g(z_i)$. We now pick some leaves $x'_1,\dots ,x'_d$ which are respective descendants of $w_1,\dots, w_d$. Observe that $x,x'_1,\dots ,x'_d$ is a clique of $R(T, g)$.
\end{proof}

\begin{theorem}\label{th:inddec}
If $\cal C$ is \pcb~with function $\chi \leq \omega^k$, then ${\cal C}_i$ is \pcb~with function $\chi \leq \omega^{k+1}$
\end{theorem}

\begin{proof}
Let $G \in \mathcal{C}_i$. There exists an independent $\cal C$-tree-decomposition $(T, g)$ such that $G = R(T, g)$. By Lemma \ref{lem:depth}, we have $\omega(G) \geq d(T, g) + 1$. For every internal node $x$, we have $g(x) \in \cal C$ and $\omega(g(x)) \leq \omega(G)$, thus $g(x)$ is $\omega(G)^k$-colorable. For every internal node $x$, let $\alpha_x$ be an $\omega(G)^k$-coloring of $g(x)$. Let $v$ be any vertex of $G$ (equivalently $v$ is a leaf of $T$). If $T$ has only one node, then $G$ is the graph on a single vertex, so $G$ is $\omega(G)^{k+1}$-colorable. Otherwise, let $x$ be the parent of $v$ in $T$. We set $c(v) = (d(v), \alpha_x(v))$, where $d(v)$ is the depth of $v$ in $(T, g)$. \\
We claim that $c$ is a proper $\omega(G)^{k+1}$-coloring of $G$. First, for every vertex $v$ of $G$, we have $0 \leq d(v) \leq d(T, g) \leq \omega(G) - 1$, so we indeed use at most $\omega(G)^{k+1}$ colors. Then, let $uv$ be an edge of $G$. Let $z$ be the closest ancestor of $u$ and $v$ in $T$, and $u', v'$ be the respective children of $z$ which are the ancestors of $u, v$. The edge $u'v'$ belongs to  $g(z)$. Since $(T, g)$ is an independent $\cal C$-tree-decomposition, this implies that either $u'$ or $v'$ is a leaf (or both). If they are both leaves, we have $u = u'$ and $v = v'$, therefore $u$ and $v$ are children of $z$ which are adjacent in $g(z)$. Thus, $\alpha_z(u) \neq \alpha_z(v)$ so $c(u) \neq c(v)$. If only one of $u', v'$ is a leaf, say $u'$, then we have that $v'$ is a strict ancestor of $v$ which is not isolated, so $d(v) > d(v') = d(u') = d(u)$ hence $c(u) \neq c(v)$. 
\end{proof}
\begin{proof}[Proof of Theorem \ref{th:chibordec}. ] The proof is by induction on the depth of $G$. \\
If $G$ has depth 0, then $G$ is a disjoint union of graphs of $\cal C$, so $\chi(G) \leq \omega(G)^k \leq \omega(G)^{2k+3}$.
Now, suppose $G$ has depth 1. By a previous lemma, we have $\omega(G) \geq 2$. We can assume that $G$ is connected since we can deal with each connected component separately. Then $G$ can be obtained by starting from a graph $G_0 \in \mathcal{C}$ on vertex set $v_1, \ldots, v_n$ and substituting each $v_i$ by a graph $G_i \in \mathcal{C}$. Let $X = \{v_i, \omega(G_i) \leq \sqrt{\omega(G)}\}$ and $Y = V(G_0) \setminus X$. Let $G_X$ be the subgraph of $G$ induced by the leaves which are descendant of nodes of $X$ and $G_Y$ be the graph induced on the other vertices. The graph $G_X$ is obtained by substituting graphs of $\mathcal{C}$ of clique number at most $\sqrt{\omega(G)}$ inside a graph of $\mathcal{C}$ of clique number at most $\omega(G)$. Hence, $G_X$ can be colored with at most  $\omega(G)^{1.5k}$ colors (consider the product of the colorings). Similarly, the graph $G_Y$ is obtained by substituting graphs of $\mathcal{C}$ of clique number at most $\omega(G)$ inside a graph of $\mathcal{C}$ of clique number at most $\sqrt{\omega(G)}$ so it can be colored with at most  $\omega(G)^{1.5k}$ colors. Hence $G$ can be colored with $2\omega(G)^{1.5k} \leq \omega(G)^{1.5k+1} \leq \omega(G)^{2k+3}$ colors. \\

Now, suppose that $G$ has depth at least 2. Thus, $\omega(G) \geq 3$. Let $(T, g)$ be a $\mathcal{C}$-tree-decomposition of $G$. For every node $x$ of $T$, let $T_x$ be the subtree of $T$ that is rooted in $x$, $g_x$ be the restriction of $g$ to the internal nodes of $T_x$, and $G_x = R(T_x, g_x)$. Finally, let $X = \{x, \omega(G_x) > \omega(G)/2\}$. Note that if $z$ is the parent of $x$ in $T$, then $G_x$ is an induced subgraph of $G_z$, so if $x \in X$, then we also have $z \in X$. In particular, $T[X]$ is a subtree of $T$. Let $C(X)$ be the set of children of an element of $X$ in $T$ and $Y = X \cup C(X)$. Let $T' = T[Y]$. $T'$ is also a subtree of $T$. Let $g'$ be the restriction of $g$ to the internal nodes of $T'$, and let $H = R(T', g')$. Note that $H$ is an induced subgraph of $G$ so $\omega(H) \leq \omega(G)$. Now, suppose $x, y$ are internal nodes of $T'$ with the same parent $z$. Then, $x, y \in X$ so $\omega(G_x), \omega(G_y) > \omega(G)/2$. Thus, $x$ and $y$ cannot be adjacent in $g'(z)$ otherwise we would have $\omega(H) > \omega(G)$. This means that $(T', g')$ is an independent $\mathcal{C}$-tree-decomposition of $H$, so $H \in \mathcal{C}_i$. Thus, $H$ is $\omega(H)^{k+1} \leq \omega(G)^{k+1}$-colorable. \\
Let $v_1, \ldots, v_n$ be the vertices of $H$. We have that $G$ can be obtained by substituting each $v_i$ by some $G_{v_i} \in \mathcal{C}_s$. We can assume that $G$ is connected since we can deal with each connected component separately. Thus, we get that for every $i \in [n]$, $d(G_{v_i}) < d(G)$ (where $d$ stands for the depth). Futhermore, if $v_i$ is a vertex of $H$, it means that it is a leaf of $T'$. Indeed, since $\omega(G) \geq 3$, if $v_i \in X$ then $\omega(G_{v_i}) \geq 2$ so $v_i$ has children in $T$. Thus, $v_i \notin X$, which means $\omega(G_{v_i}) \leq \omega(G)/2$. Let $\omega_i = \omega(G_{v_i})$. By induction hypothesis, we have $\chi(G_{v_i}) \leq \omega_i^{2k+3}$. \\
Let $m = \lfloor\log(\omega(G))\rfloor$. For $j \in [m]$, let $V_j = \{v_i: \frac{\omega(G)}{2^{j+1}} < \omega_i \leq \frac{\omega(G)}{2^j}\}$. Note that the $V_j$'s partition  $\{v_1, \ldots, v_n\}$. For $j \in [m]$, let $H_j = H[V_j]$ and $G_j$ be the corresponding induced subgraph of $G$. 

If $j \in [m]$ and $v_i \in V_j$ then $\omega_i > \frac{\omega(G)}{2^{j+1}}$ so $\omega(H_j) \leq 2^{j+1}$. Thus, $\chi(H_j) \leq 2^{(j+1)(k+1)}$ because $H_j$ is an induced sugraph of $H\in \mathcal{C}_i$. Thus, we have $\chi(G_j) \leq 2^{(j+1)(k+1)} \cdot \left(\frac{\omega(G)}{2^{j}}\right)^{2k+3}$ (by doing the product of the colorings). Now, let's color separately each $V_j$ with a different palette. The number of colors used is: \begin{align*}
    \sum_{j=1}^m \chi(G_j) &\leq \sum_{j = 1}^m 2^{(j+1)(k+1)} \cdot \left(\frac{\omega(G)}{2^{j}}\right)^{2k+3} \\
                            &= 2^{k+1} \omega(G)^{2k+3} \sum_{j = 1}^m 2^{-j(k+2)} \\
                            &\leq 2^{k+1} \omega(G)^{2k+3} \frac{2^{-(k+2)}}{1-2^{-(k+2)}} \\
                            &\leq \omega(G)^{2k+3} \times \frac{2^{k+2}}{2(2^{k+2} - 1)} \\
                            &\leq \omega(G)^{2k+3}.
\end{align*}
\end{proof}
\section{Mixed extensions}

We conclude with another type of class extension which preserves  $\chi$-boundedness. Given an ordered graph $G$ with vertex ordering  $v_1,\dots ,v_n$ (denoted by $<$), we only ask here that the restriction to "mixed parts" of $G$ are $\chi$-bounded. Precisely, given a class $\cal C$ of graphs, we say that $(G,<)$ is a \emph{$\cal C$-mixed extension} if for every partition $\cal P$ into intervals $I_1,\dots ,I_k$, the subgraph $\text{Mix}(G,<,\cal P)$ of $G$ in which we only keep edges between intervals $I_s,I_t$ for which $I_s,I_t$ is mixed, belongs to $\cal C$. The class of graphs $G$ admitting a vertex-ordering $<$ such that $(G,<)$ is a $\cal C$-mixed extension is denoted by ME$(\cal C)$.

Observe that graphs with twin-width at most $d$ are mixed extensions of the class of graphs with chromatic number at most $f(d)$. To see this, observe that they admit vertex-orderings with no large mixed-minors, hence the Marcus-Tardos theorem implies that every partition $\cal P$ has a linear number of mixed zones. So $\text{Mix}(G,<,\cal P)$ has bounded chromatic number by degeneracy. Hence the fact that bounded twin-width graphs are $\chi$-bounded is a particular case of the following result (whose proof mimics the one of Proposition~\ref{prop:RMPpcb}):

\begin{theorem}\label{th:chimix}
If $\cal C$ is $\chi$-bounded, then ME$(\cal C)$ is $\chi$-bounded.
\end{theorem}

\begin{proof} Assume that all graphs $H$ in $\mathcal{C}$ satisfy $\chi(H)\leq f(\omega (H))$. Consider $G\in \cal C$ with clique number $\omega$ and assume that we have already proved that the chromatic number of every graph in $ME(\mathcal{C})$ with smaller value of $\omega$ is at most $c$. Consider $<$, a vertex ordering $v_1,\dots ,v_n$ of $G$ such that $(G,<)$ is a $\cal C$-mixed extension. Pick a smallest initial interval $I_1=v_1,\dots ,v_{i_1}$ such that $G[I_1]$ has a clique of size $\omega$. Then pick a smallest interval $I_2=v_{i_1+1},\dots ,v_{i_2}$ such that $G[I_2]$ has a clique of size $\omega$, otherwise pick all remaining vertices. Continue the process to get a partition $\cal P$ into intervals $I_1,\dots ,I_k$. Observe that if there is an edge between $I_s$ and $I_t$, where $s<t<k$, then $I_s,I_t$ is mixed, otherwise one of the intervals would be a module with respect to the other, and we would find a clique of size $\omega +1$. Therefore the chromatic number of $G$ is at most $(c+1)f(\omega )+c+1$ since $\chi(G[I_k])\leq c+1$ and  $G[I_1\cup \dots \cup I_{k-1}]$ is the edge union of two graphs, one of chromatic number at most $c+1$ inside the $I_i$'s, and one of chromatic number at most $f(\omega)$ across the $I_i$'s.
\end{proof}


\bibliographystyle{amsplain}

\begin{thebibliography}{10}

\bibitem{BP19}
Marthe Bonamy and Micha\l{} Pilipczuk.
\newblock Graphs of bounded cliquewidth are polynomially $\chi$-bounded.
\newblock {\em Advances in Combinatorics}, (8):21pp, 2020.

\bibitem{BBDGT23}
\'Edouard Bonnet, Romain Bourneuf, Julien Duron, Colin Geniet, Stéphan
  Thomassé, and Nicolas Trotignon.
\newblock A tamed family of triangle-free graphs with unbounded chromatic
  number.
\newblock 2023.

\bibitem{BGKTW21}
\'{E}douard Bonnet, Colin Geniet, Eun~Jung Kim, St\'{e}phan Thomass\'{e}, and
  R\'{e}mi Watrigant.
\newblock {Twin-width III: Max Independent Set, Min Dominating Set, and
  Coloring}.
\newblock In Nikhil Bansal, Emanuela Merelli, and James Worrell, editors, {\em
  48th International Colloquium on Automata, Languages, and Programming (ICALP
  2021)}, volume 198 of {\em Leibniz International Proceedings in Informatics
  (LIPIcs)}, pages 35:1--35:20, Dagstuhl, Germany, 2021. Schloss Dagstuhl --
  Leibniz-Zentrum f{\"u}r Informatik.

\bibitem{BKTW21}
\'{E}douard Bonnet, Eun~Jung Kim, St\'{e}phan Thomass\'{e}, and R\'{e}mi
  Watrigant.
\newblock Twin-width {I}: Tractable {FO} model checking.
\newblock {\em J. ACM}, 69(1), 2021.

\bibitem{BDW22}
Marcin Briański, James Davies, and Bartosz Walczak.
\newblock Separating polynomial $\chi$-boundedness from $\chi$-boundedness.
\newblock {\em Combinatorica}, 44(1):1--8, 2023.

\bibitem{CHMS23}
Alvaro Carbonero, Patrick Hompe, Benjamin Moore, and Sophie Spirkl.
\newblock A counterexample to a conjecture about triangle-free induced
  subgraphs of graphs with large chromatic number.
\newblock {\em Journal of Combinatorial Theory, Series B}, 158:63--69, 
  2023.

\bibitem{CPST13}
Maria Chudnovsky, Irena Penev, Alex Scott, and Nicolas Trotignon.
\newblock Substitution and $\chi$-boundedness.
\newblock {\em Journal of Combinatorial Theory, Series B}, 103(5):567--586,
  2013.

\bibitem{CRST06}
Maria Chudnovsky, Neil Robertson, Paul Seymour, and Robin Thomas.
\newblock The strong perfect graph theorem.
\newblock {\em Annals of Mathematics}, 164:51--229, 2006.

\bibitem{DM21}
James Davies and Rose McCarty.
\newblock Circle graphs are quadratically $\chi$-bounded.
\newblock {\em Bulletin of the London Mathematical Society}, 53(3):673--679,
  2021.

\bibitem{BD54}
Blanche Descartes.
\newblock Solution to advanced problem no. 4526.
\newblock {\em Amer. Math. Monthly}, 61, 1954.

\bibitem{DK12}
Zdeněk Dvořák and Daniel Král’.
\newblock Classes of graphs with small rank decompositions are $\chi$-bounded.
\newblock {\em European Journal of Combinatorics}, 33(4):679--683, 2012.

\bibitem{ER59}
Paul Erd\H{o}s.
\newblock Graph theory and probability.
\newblock {\em Canad. J. Math.}, 11:34--38, 1959.

\bibitem{EH64}
Paul Erdős and A.~Hajnal.
\newblock {Some remarks on set theory. IX. Combinatorial problems in measure
  theory and set theory.}
\newblock {\em Michigan Mathematical Journal}, 11(2):107--127, 1964.

\bibitem{GIPSSTT24}
Ant{\'{o}}nio Gir{\~{a}}o, Freddie Illingworth, Emil Powierski, Michael Savery,
  Alex Scott, Youri Tamitegama, and Jane Tan.
\newblock Induced subgraphs of induced subgraphs of large chromatic number.
\newblock {\em Comb.}, 44(1):37--62, 2024.

\bibitem{LSWY21}
Xiaonan Liu, Joshua Schroeder, Zhiyu Wang, and Xingxing Yu.
\newblock Polynomial $\chi$-binding functions for t-broom-free graphs.
\newblock {\em Journal of Combinatorial Theory, Series B}, 162:118--133, 2023.

\bibitem{MT04}
Adam Marcus and G{\'a}bor Tardos.
\newblock Excluded permutation matrices and the {S}tanley-{W}ilf conjecture.
\newblock {\em Journal of Combinatorial Theory - Series A}, 107(1):153--160,
  2004.

\bibitem{MY55}
Jan Mycielski.
\newblock Sur le coloriage des graphes.
\newblock {\em Colloq. Math.}, 3:161--162, 1955.

\bibitem{PS22}
Micha\l{} Pilipczuk and Marek Soko\l{}owski.
\newblock Graphs of bounded twin-width are quasi-polynomially $\chi$-bounded.
\newblock {\em J. Comb. Theory Ser. B}, 161(C):382--406, 2023.

\bibitem{SS20}
Alex Scott and Paul Seymour.
\newblock A survey of $\chi$-boundedness.
\newblock {\em Journal of Graph Theory}, 95(3):473--504, 2020.

\bibitem{SSS23}
Alex Scott, Paul Seymour, and Sophie Spirkl.
\newblock Polynomial bounds for chromatic number {I}. {E}xcluding a biclique
  and an induced tree.
\newblock {\em Journal of Graph Theory}, 102(3):458--471, 2023.

\bibitem{SSS22}
Alex Scott, Paul Seymour, and Sophie Spirkl.
\newblock Polynomial bounds for chromatic number. {IV}: A near-polynomial bound
  for excluding the five-vertex path.
\newblock {\em Combinatorica}, 43(5):845--852, 2023.

\bibitem{ZY52}
Alexander~A. Zykov.
\newblock On some properties of linear complexes.
\newblock {\em Amer. Math. Soc. Transl.}, 79:163--188, 1952.

\end{thebibliography}


\begin{aicauthors}
\begin{authorinfo}[rom]
  Romain Bourneuf\\
  Univ. Bordeaux, CNRS, Bordeaux INP, LaBRI, UMR 5800, F-33400\\
  Talence, France\\
  romain\imagedot{}bourneuf\imageat{}ens-lyon\imagedot{}fr \\
  \url{https://perso.ens-lyon.fr/romain.bourneuf/}
\end{authorinfo}
\begin{authorinfo}[stef]
  Stéphan Thomassé\\
  Univ Lyon, EnsL, UCBL, CNRS, LIP, F-69342\\
  LYON Cedex 07, France\\
  stephan\imagedot{}thomasse\imageat{}ens-lyon\imagedot{}fr \\
  \url{https://perso.ens-lyon.fr/stephan.thomasse/}
\end{authorinfo}
\end{aicauthors}

\end{document}